\documentclass[12pt]{article}

\usepackage{graphicx, amssymb, latexsym, amsfonts, amsmath, lscape, amscd,
amsthm, color, epsfig, mathrsfs, tikz, enumerate}
\usetikzlibrary{decorations.markings}

\newtheorem{theorem}{Theorem}[section]

\newtheorem{lemma}[theorem]{Lemma}

\newcommand\DELETE[1]{}

\begin{document}

\title{{\bf Erratum for ``On oriented cliques with respect to push operation''}}
\author{
{\sc Julien Bensmail}$\,^a$, {\sc Soumen Nandi}$\,^{b}$, {\sc Sagnik Sen$^{c}$}\\
\mbox{}\\
{\small $(a)$ Universit\'e C\^ote d'Azur, Inria, CNRS, I3S, France}\\
{\small $(b)$ Birla Institute of Technology \& Science Pilani, Hyderabad Campus, India}\\
{\small $(c)$ Ramakrishna Mission Vivekananda Educational and Research Institute, Belur Math, India}}


\date{\today}

\maketitle

\begin{abstract}
An error is spotted in the statement of  Theorem~1.3 of our published article titled 
  ``On oriented cliques with respect to push operation'' (Discrete Applied Mathematics 2017). 
  The theorem provided an exhaustive list of 16 minimal (up to spanning subgraph inclusion)
 underlying planar push cliques. The error was that, one of the 16 graphs from the above list was missing an arc. 
  We correct the error and restate the corrected statement in this article. We also point out the reason for the error and comment that the error occurred due to a mistake in a particular lemma. We present the corrected proof of that particular lemma as well. Moreover, a few counts were wrongly reported due to the above mentioned error. So we update our reported counts after correction in this article. 
\end{abstract}



\section{Introduction and main results}
This is an erratum reporting and correcting an error in our 
published article 
titled  ``On oriented cliques with respect to push operation''~\cite{bensmail2017oriented}. 
We spotted the error after the article was published. All notations and terminologies used in this article are the same as the ones used in the above mentioned paper~\cite{bensmail2017oriented}.

The error occured in Theorem~1.3 of the concerned article~\cite{bensmail2017oriented}.
The theorem characterizes all minimal (up to spanning subgraph inclusion)
 underlying planar push cliques. 
 The theorem claimed that there are exactly $16$ such graphs and 
 we provided a figure (Figure~1 in our paper~\cite{bensmail2017oriented}) which 
 illustrated all the $16$ graphs. 
 In the figure, we actually presented orientations of those $16$ undirected graphs 
 in such a way that they were push cliques.  
We  called those graphs as 
$\overrightarrow{H}_1, \overrightarrow{H}_2, \cdots, \overrightarrow{H}_{16}$.

However, our list had an error. 
The graph $\overrightarrow{H}_{10}$, as illustrated 
in Figure~1 of our paper~\cite{bensmail2017oriented} is not a push clique. 
After finding the error, we rechecked our proof and found that the error can be fixed easily by just adding one arc in $\overrightarrow{H}_{10}$. 
Moreover, our proof is correct throughout, except for Lemma~5.9 of the 
paper~\cite{bensmail2017oriented}.

Even though the error can be corrected by simply adding an arc in the figure, it is a significant error as anyone using our result may create an error in their works as well. So we feel compelled to present the corrected version of 
Theorem~1.3 from~\cite{bensmail2017oriented} here. There are some errors in the proof of Lemma~5.9 as well, mainly due to the error in the figure. However, the errors are 
restricted to the ``case analysis'' portion of the proof and can be fixed without too much trouble. We will also provide a short proof of Lemma~5.9 here.

Moreover, the total number of non-isomorphic  planar underlying push cliques reported in the Section~6 (Conclusions) of the paper~\cite{bensmail2017oriented} will also be updated. 
We are going to report the updated count in the last paragraph of this article.

Therefore, we  present a corrected version of Theorem~1.3 from~\cite{bensmail2017oriented} in the following.

\begin{figure}

\centering
\begin{tikzpicture}[scale=.8]

\filldraw [black] (-.5,6) circle (2pt) {node[left]{}};

\node at (-.5,4.25) {$\overrightarrow{H}_1$};

\filldraw [black] (.5,6.5) circle (2pt) {node[above]{}};
\filldraw [black] (.5,5.5) circle (2pt) {node[below left]{}};

\begin{scope}[thin,decoration={
    markings,
    mark=at position 0.5 with {\arrow{>}}}
    ] 
\draw[postaction={decorate}]  (.5,5.5) -- (.5,6.5);
\end{scope}

\node at (.5,4.25) {$\overrightarrow{H}_2$};

\filldraw [black] (2,7) circle (2pt) {node[above]{}};
\filldraw [black] (1.5,6) circle (2pt) {node[below left]{}};
\filldraw [black] (2,5) circle (2pt) {node[below left]{}};

\begin{scope}[thin,decoration={
    markings,
    mark=at position 0.5 with {\arrow{>}}}
    ] 
\draw[postaction={decorate}] (2,7) -- (1.5,6);
\draw[postaction={decorate}] (1.5,6) -- (2,5);
\draw[postaction={decorate}] (2,5) -- (2,7);
\end{scope}

\node at (1.7,4.25) {$\overrightarrow{H}_3$};

\filldraw [black] (3,6.5) circle (2pt) {node[above]{}};
\filldraw [black] (3,5.5) circle (2pt) {node[below left]{}};
\filldraw [black] (4,6.5) circle (2pt) {node[below left]{}};
\filldraw [black] (4,5.5) circle (2pt) {node[below left]{}};

\begin{scope}[thin,decoration={
    markings,
    mark=at position 0.5 with {\arrow{>}}}
    ] 
\draw[postaction={decorate}] (3,6.5) -- (3,5.5);
\draw[postaction={decorate}] (4,6.5) -- (4,5.5);
\draw[postaction={decorate}] (3,6.5) -- (4,6.5);
\draw[postaction={decorate}] (4,5.5) -- (3,5.5);
\end{scope}

\node at (3.5,4.25) {$\overrightarrow{H}_4$};

\filldraw [black] (5,5.25) circle (2pt) {node[above]{}};
\filldraw [black] (7,5.25) circle (2pt) {node[below left]{}};
\filldraw [black] (5,6.25) circle (2pt) {node[below left]{}};
\filldraw [black] (7,6.25) circle (2pt) {node[below left]{}};
\filldraw [black] (6,7.25) circle (2pt) {node[below left]{}};

\begin{scope}[thin,decoration={
    markings,
    mark=at position 0.5 with {\arrow{>}}}
    ] 
\draw[postaction={decorate}] (5,5.25) -- (7,5.25);
\draw[postaction={decorate}] (7,5.25) -- (7,6.25);
\draw[postaction={decorate}] (7,6.25) -- (6,7.25);
\draw[postaction={decorate}] (5,6.25) -- (6,7.25);
\draw[postaction={decorate}] (5,5.25) -- (5,6.25);
\draw[postaction={decorate}] (7,6.25) -- (5,6.25);
\draw[postaction={decorate}] (5,5.25) .. controls (4.15,5.95) and (4.5,7.05) .. (6,7.25);
\end{scope}

\node at (6,4.25) {$\overrightarrow{H}_5$};

\filldraw [black] (9,7.5) circle (2pt) {node[above]{}};
\filldraw [black] (8,6.75) circle (2pt) {node[above]{}};
\filldraw [black] (10,6.75) circle (2pt) {node[above]{}};
\filldraw [black] (8,5.75) circle (2pt) {node[above]{}};
\filldraw [black] (10,5.75) circle (2pt) {node[above]{}};
\filldraw [black] (9,5) circle (2pt) {node[above]{}};

\begin{scope}[thin,decoration={
    markings,
    mark=at position 0.5 with {\arrow{>}}}
    ] 
\draw[postaction={decorate}] (8,6.75) -- (9,7.5);
\draw[postaction={decorate}] (8,6.75) -- (8,5.75);
\draw[postaction={decorate}] (8,5.75) -- (9,5);
\draw[postaction={decorate}] (9,5) -- (10,5.75);
\draw[postaction={decorate}] (10,5.75) -- (10,6.75);
\draw[postaction={decorate}] (10,6.75) -- (9,7.5);
\draw[postaction={decorate}] (8,5.75) -- (10,5.75);
\draw[postaction={decorate}] (8,5.75) -- (10,6.75);
\draw[postaction={decorate}] (10, 5.75) .. controls (10.2,5.75) and (11,7.25) .. (9,7.5);
\draw[postaction={decorate}] (9,5) .. controls (11.25,5.5) and (11.25,7.5) .. (9,7.5);
\end{scope}

\node at (9,4.25) {$\overrightarrow{H}_6$};

\filldraw [black] (1,3.5) circle (2pt) {node[above]{}};
\filldraw [black] (0,2.5) circle (2pt) {node[above]{}};
\filldraw [black] (2,2.5) circle (2pt) {node[above]{}};
\filldraw [black] (0,1.5) circle (2pt) {node[above]{}};
\filldraw [black] (2,1.5) circle (2pt) {node[above]{}};
\filldraw [black] (1,.5) circle (2pt) {node[above]{}};

\begin{scope}[thin,decoration={
    markings,
    mark=at position 0.5 with {\arrow{>}}}
    ] 
\draw[postaction={decorate}] (1,3.5) -- (0,2.5);
\draw[postaction={decorate}] (0,1.5) -- (0,2.5);
\draw[postaction={decorate}] (0,1.5) -- (1,.5);
\draw[postaction={decorate}] (2,1.5) -- (1,.5);
\draw[postaction={decorate}] (2,1.5) -- (2,2.5);
\draw[postaction={decorate}] (1,3.5) -- (2,2.5);
\draw[postaction={decorate}] (2,2.5) -- (0,2.5);
\draw[postaction={decorate}] (0,1.5) -- (2,1.5);
\draw[postaction={decorate}] (1,.5) .. controls (3,1) and (3,3) .. (1,3.5);
\end{scope}

\node at (1,-.25) {$\overrightarrow{H}_7$};

\filldraw [black] (5,3.5) circle (2pt) {node[above]{}};
\filldraw [black] (4,2.5) circle (2pt) {node[above]{}};
\filldraw [black] (6,2.5) circle (2pt) {node[above]{}};
\filldraw [black] (4,1.5) circle (2pt) {node[above]{}};
\filldraw [black] (6,1.5) circle (2pt) {node[above]{}};
\filldraw [black] (5,.5) circle (2pt) {node[above]{}};

\begin{scope}[thin,decoration={
    markings,
    mark=at position 0.5 with {\arrow{>}}}
    ] 
\draw[postaction={decorate}] (5,3.5) -- (4,2.5);
\draw[postaction={decorate}] (4,1.5) -- (4,2.5);
\draw[postaction={decorate}] (5,.5) -- (4,1.5);
\draw[postaction={decorate}] (5,.5) -- (6,1.5);
\draw[postaction={decorate}] (6,1.5) -- (6,2.5);
\draw[postaction={decorate}] (5,3.5) -- (6,2.5);

\draw[postaction={decorate}] (6,1.5) -- (5,3.5);
\draw[postaction={decorate}] (5,3.5) -- (4,1.5);
\draw[postaction={decorate}] (5,3.5) -- (5,.5);
\draw[postaction={decorate}] (4,2.5) .. controls (4.3,4.35) and (5.7,4.35) .. (6,2.5);
\end{scope}

\node at (5,-.25) {$\overrightarrow{H}_8$};

\filldraw [black] (9,3.5) circle (2pt) {node[above]{}};
\filldraw [black] (8,2.5) circle (2pt) {node[above]{}};
\filldraw [black] (10,2.5) circle (2pt) {node[above]{}};
\filldraw [black] (8,1.5) circle (2pt) {node[above]{}};
\filldraw [black] (10,1.5) circle (2pt) {node[above]{}};
\filldraw [black] (9,.5) circle (2pt) {node[above]{}};

\begin{scope}[thin,decoration={
    markings,
    mark=at position 0.5 with {\arrow{>}}}
    ] 
\draw[postaction={decorate}] (8,2.5) -- (9,3.5);
\draw[postaction={decorate}] (8,1.5) -- (8,2.5);
\draw[postaction={decorate}] (8,1.5) -- (9,.5);
\draw[postaction={decorate}] (9,.5) -- (10,1.5);
\draw[postaction={decorate}] (10,1.5) -- (10,2.5);
\draw[postaction={decorate}] (10,2.5) -- (9,3.5);

\draw[postaction={decorate}] (8,1.5) -- (10,1.5);
\draw[postaction={decorate}] (10,2.5) -- (8,2.5);
\draw[postaction={decorate}] (8,1.5) .. controls (7,2.7) and (8,3.5) .. (9,3.5);
\draw[postaction={decorate}] (9,.5) .. controls (10,.7) and (11.18,1.15) .. (10,2.5);
\end{scope}

\node at (9,-.25) {$\overrightarrow{H}_9$};

\filldraw [black] (1,-1) circle (2pt) {node[above]{}};
\filldraw [black] (0,-2) circle (2pt) {node[above]{}};
\filldraw [black] (2,-2) circle (2pt) {node[above]{}};
\filldraw [black] (0,-3) circle (2pt) {node[above]{}};
\filldraw [black] (2,-3) circle (2pt) {node[above]{}};
\filldraw [black] (.5,-4) circle (2pt) {node[above]{}};
\filldraw [black] (1.5,-4) circle (2pt) {node[above]{}};

\begin{scope}[thin,decoration={
    markings,
    mark=at position 0.5 with {\arrow{>}}}
    ] 
\draw[postaction={decorate}] (1,-1) -- (0,-2);
\draw[postaction={decorate}] (0,-2) -- (0,-3);
\draw[postaction={decorate}] (2,-3) -- (1,-1);
\draw[postaction={decorate}] (0,-3) -- (.5,-4);
\draw[postaction={decorate}] (1.5,-4) -- (.5,-4);
\draw[postaction={decorate}] (1.5,-4) -- (2,-3);
\draw[postaction={decorate}] (2,-3) -- (2,-2);
\draw[postaction={decorate}] (2,-2) -- (1,-1);

\draw[postaction={decorate}] (2,-3) -- (0,-2);
\draw[postaction={decorate}] (1.5,-4) -- (0,-3);
\draw[postaction={decorate}] (1.5,-4) .. controls (2.3,-3.8) and (2.7,-2.6) .. (2,-2);
\draw[postaction={decorate}] (1,-1) .. controls (-.5,-1.4) and (-.3,-2.5) .. (0,-3);
\draw[postaction={decorate}] (1,-1) .. controls (-1.3,-1.2) and (-1,-3.8) .. (.5,-4);
\end{scope}

\node at (1,-4.75) {$\overrightarrow{H}_{10}$};

\filldraw [black] (5,-1) circle (2pt) {node[above]{}};
\filldraw [black] (4,-2) circle (2pt) {node[above]{}};
\filldraw [black] (6,-2) circle (2pt) {node[above]{}};
\filldraw [black] (4,-3) circle (2pt) {node[above]{}};
\filldraw [black] (6,-3) circle (2pt) {node[above]{}};
\filldraw [black] (4.5,-4) circle (2pt) {node[above]{}};
\filldraw [black] (5.5,-4) circle (2pt) {node[above]{}};

\begin{scope}[thin,decoration={
    markings,
    mark=at position 0.5 with {\arrow{>}}}
    ] 
\draw[postaction={decorate}] (4,-2) -- (5,-1);
\draw[postaction={decorate}] (4,-2) -- (4,-3);
\draw[postaction={decorate}] (4.5,-4) -- (4,-3);
\draw[postaction={decorate}] (5.5,-4) -- (4.5,-4);
\draw[postaction={decorate}] (6,-3) -- (5.5,-4);
\draw[postaction={decorate}] (6,-2) -- (6,-3);
\draw[postaction={decorate}] (5,-1) -- (6,-2);

\draw[postaction={decorate}] (4,-2) -- (4.5,-4);
\draw[postaction={decorate}] (4,-2) -- (6,-3);
\draw[postaction={decorate}] (4,-2) -- (6,-2);
\draw[postaction={decorate}] (6,-2) .. controls (6.7,-3) and (6.2,-3.5) .. (5.5,-4);
\draw[postaction={decorate}] (5.5,-4) .. controls (5,-4.2) and (4.1,-5) .. (4,-3);
\draw[postaction={decorate}] (5,-1) .. controls (7.2,-1.75) and (7,-3.5) .. (5.5,-4);
\end{scope}

\node at (5,-4.75) {$\overrightarrow{H}_{11}$};

\filldraw [black] (9,-1) circle (2pt) {node[above]{}};
\filldraw [black] (8,-2) circle (2pt) {node[above]{}};
\filldraw [black] (10,-2) circle (2pt) {node[above]{}};
\filldraw [black] (8,-3) circle (2pt) {node[above]{}};
\filldraw [black] (10,-3) circle (2pt) {node[above]{}};
\filldraw [black] (8.5,-4) circle (2pt) {node[above]{}};
\filldraw [black] (9.5,-4) circle (2pt) {node[above]{}};

\begin{scope}[thin,decoration={
    markings,
    mark=at position 0.5 with {\arrow{>}}}
    ] 
\draw[postaction={decorate}] (8,-2) -- (9,-1);
\draw[postaction={decorate}] (8,-2) -- (8,-3);
\draw[postaction={decorate}] (8,-3) -- (8.5,-4);
\draw[postaction={decorate}] (8.5,-4) -- (9.5,-4);
\draw[postaction={decorate}] (9.5,-4) -- (10,-3);
\draw[postaction={decorate}] (10,-3) -- (10,-2);
\draw[postaction={decorate}] (9,-1) -- (10,-2);

\draw[postaction={decorate}] (8.5,-4) -- (9,-1);
\draw[postaction={decorate}] (9.5,-4) -- (9,-1);
\draw[postaction={decorate}] (10,-2) -- (9.5,-4);
\draw[postaction={decorate}] (8,-2) .. controls (8.6,-.3) and (9.4,-.3) .. (10,-2);
\draw[postaction={decorate}] (8,-3) .. controls (8,-5) and (10,-5) .. (10,-3);
\end{scope}

\node at (9,-4.95) {$\overrightarrow{H}_{12}$};

\filldraw [black] (-1.5,-5.5) circle (2pt) {node[above]{}};
\filldraw [black] (-.5,-5.5) circle (2pt) {node[above]{}};
\filldraw [black] (-2,-6.5) circle (2pt) {node[above]{}};
\filldraw [black] (0,-6.5) circle (2pt) {node[above]{}};
\filldraw [black] (-2,-7.5) circle (2pt) {node[above]{}};
\filldraw [black] (0,-7.5) circle (2pt) {node[above]{}};
\filldraw [black] (-1.5,-8.5) circle (2pt) {node[above]{}};
\filldraw [black] (-.5,-8.5) circle (2pt) {node[above]{}};

\begin{scope}[thin,decoration={
    markings,
    mark=at position 0.5 with {\arrow{>}}}
    ] 
\draw[postaction={decorate}] (-2,-6.5) -- (-1.5,-5.5);
\draw[postaction={decorate}] (-2,-6.5) -- (-2,-7.5);
\draw[postaction={decorate}] (-2,-7.5) -- (-1.5,-8.5);
\draw[postaction={decorate}] (-.5,-8.5) -- (-1.5,-8.5);
\draw[postaction={decorate}] (0,-7.5) -- (-.5,-8.5);
\draw[postaction={decorate}] (0,-7.5) -- (0,-6.5);
\draw[postaction={decorate}] (0,-6.5) -- (-.5,-5.5);
\draw[postaction={decorate}] (-.5,-5.5) -- (-1.5,-5.5);

\draw[postaction={decorate}] (-2,-6.5) -- (-1.5,-8.5);
\draw[postaction={decorate}] (-1.5,-8.5) -- (0,-7.5);
\draw[postaction={decorate}] (0,-7.5) -- (-.5,-5.5);
\draw[postaction={decorate}] (-2,-6.5) -- (-.5,-5.5);

\draw[postaction={decorate}] (-.5,-8.5) .. controls (.2,-8) and (.7,-7.5) .. (0,-6.5);
\draw[postaction={decorate}] (-1.5,-5.5) .. controls (-1.2,-5) and (.3,-5) .. (0,-6.5);
\draw[postaction={decorate}] (-1.5,-5.5) .. controls (-2,-5.5) and (-3,-6.3) .. (-2,-7.5);
\draw[postaction={decorate}] (-2,-7.5) .. controls (-2,-9.3) and (-1,-8.8) .. (-.5,-8.5);
\end{scope}

\node at (-1,-9.45) {$\overrightarrow{H}_{13}$};

\filldraw [black] (2.5,-5.5) circle (2pt) {node[above]{}};
\filldraw [black] (3.5,-5.5) circle (2pt) {node[above]{}};
\filldraw [black] (2,-6.5) circle (2pt) {node[above]{}};
\filldraw [black] (4,-6.5) circle (2pt) {node[above]{}};
\filldraw [black] (2,-7.5) circle (2pt) {node[above]{}};
\filldraw [black] (4,-7.5) circle (2pt) {node[above]{}};
\filldraw [black] (2.5,-8.5) circle (2pt) {node[above]{}};
\filldraw [black] (3.5,-8.5) circle (2pt) {node[above]{}};

\begin{scope}[thin,decoration={
    markings,
    mark=at position 0.5 with {\arrow{>}}}
    ] 
\draw[postaction={decorate}] (2,-6.5) -- (2.5,-5.5);
\draw[postaction={decorate}] (2,-7.5) -- (2,-6.5);
\draw[postaction={decorate}] (2.5,-8.5) -- (2,-7.5);
\draw[postaction={decorate}] (3.5,-8.5) -- (2.5,-8.5);
\draw[postaction={decorate}] (4,-7.5) -- (3.5,-8.5);
\draw[postaction={decorate}] (4,-6.5) -- (4,-7.5);
\draw[postaction={decorate}] (3.5,-5.5) -- (4,-6.5);
\draw[postaction={decorate}] (3.5,-5.5) -- (2.5,-5.5);

\draw[postaction={decorate}] (3.5,-5.5) -- (2,-6.5);
\draw[postaction={decorate}] (3.5,-5.5) -- (2,-7.5);

\draw[postaction={decorate}] (3.5,-5.5) -- (3.5,-8.5);
\draw[postaction={decorate}] (3.5,-5.5) -- (4,-7.5);

\draw[postaction={decorate}] (4,-7.5) .. controls (3.8,-9) and (3.3,-9) .. (2.5,-8.5);
\draw[postaction={decorate}] (4,-6.5) .. controls (5,-8) and (3.5,-9.9) .. (2.5,-8.5);
\draw[postaction={decorate}] (2.5,-8.5) .. controls (1.3,-7.5) and (1.8,-7) .. (2,-6.5);
\draw[postaction={decorate}] (2.5,-8.5) .. controls (.8,-7.5) and (1.2,-6) .. (2.5,-5.5);
\end{scope}

\node at (3,-9.45) {$\overrightarrow{H}_{14}$};

\filldraw [black] (6.5,-5.5) circle (2pt) {node[above]{}};
\filldraw [black] (7.5,-5.5) circle (2pt) {node[above]{}};
\filldraw [black] (6,-6.5) circle (2pt) {node[above]{}};
\filldraw [black] (8,-6.5) circle (2pt) {node[above]{}};
\filldraw [black] (6,-7.5) circle (2pt) {node[above]{}};
\filldraw [black] (8,-7.5) circle (2pt) {node[above]{}};
\filldraw [black] (6.5,-8.5) circle (2pt) {node[above]{}};
\filldraw [black] (7.5,-8.5) circle (2pt) {node[above]{}};

\begin{scope}[thin,decoration={
    markings,
    mark=at position 0.5 with {\arrow{>}}}
    ] 
\draw[postaction={decorate}] (6,-6.5) -- (6.5,-5.5);
\draw[postaction={decorate}] (6,-7.5) -- (6,-6.5);
\draw[postaction={decorate}] (6.5,-8.5) -- (6,-7.5);
\draw[postaction={decorate}] (6.5,-8.5) -- (7.5,-8.5);
\draw[postaction={decorate}] (7.5,-8.5) -- (8,-7.5);
\draw[postaction={decorate}] (8,-7.5) -- (8,-6.5);
\draw[postaction={decorate}] (8,-6.5) -- (7.5,-5.5);
\draw[postaction={decorate}] (7.5,-5.5) -- (6.5,-5.5);

\draw[postaction={decorate}] (6.5,-5.5) -- (8,-6.5);
\draw[postaction={decorate}] (8,-6.5) -- (6,-6.5);
\draw[postaction={decorate}] (6,-6.5) -- (8,-7.5);
\draw[postaction={decorate}] (8,-7.5) -- (6,-7.5);
\draw[postaction={decorate}] (6,-7.5) -- (7.5,-8.5);

\draw[postaction={decorate}] (7.5,-8.5) .. controls (8.2,-8) and (8.7,-7.5) .. (8,-6.5);
\draw[postaction={decorate}] (6.5,-5.5) .. controls (6,-5.5) and (5,-6.3) .. (6,-7.5);
\draw[postaction={decorate}] (6.5,-5.5) .. controls (5,-5.5) and (4.5,-7.5) .. (6.5,-8.5);
\draw[postaction={decorate}] (7.5,-5.5) .. controls (9,-6.5) and (9.2,-7.9) .. (7.5,-8.5);
\end{scope}

\node at (7,-9.25) {$\overrightarrow{H}_{15}$};

\filldraw [black] (10.5,-5.5) circle (2pt) {node[above]{}};
\filldraw [black] (11.5,-5.5) circle (2pt) {node[above]{}};
\filldraw [black] (10,-6.5) circle (2pt) {node[above]{}};
\filldraw [black] (12,-6.5) circle (2pt) {node[above]{}};
\filldraw [black] (10,-7.5) circle (2pt) {node[above]{}};
\filldraw [black] (12,-7.5) circle (2pt) {node[above]{}};
\filldraw [black] (10.5,-8.5) circle (2pt) {node[above]{}};
\filldraw [black] (11.5,-8.5) circle (2pt) {node[above]{}};

\begin{scope}[thin,decoration={
    markings,
    mark=at position 0.5 with {\arrow{>}}}
    ] 
\draw[postaction={decorate}] (10.5,-5.5) -- (10,-6.5);
\draw[postaction={decorate}] (10,-7.5) -- (10,-6.5);
\draw[postaction={decorate}] (10,-7.5) -- (10.5,-8.5);
\draw[postaction={decorate}] (10.5,-8.5) -- (11.5,-8.5);
\draw[postaction={decorate}] (12,-7.5) -- (11.5,-8.5);
\draw[postaction={decorate}] (12,-7.5) -- (12,-6.5);
\draw[postaction={decorate}] (12,-6.5) -- (11.5,-5.5);
\draw[postaction={decorate}] (10.5,-5.5) -- (11.5,-5.5);

\draw[postaction={decorate}] (10.5,-5.5) -- (12,-6.5);
\draw[postaction={decorate}] (10.5,-5.5) -- (11.5,-8.5);
\draw[postaction={decorate}] (10,-7.5) -- (11.5,-8.5);

\draw[postaction={decorate}] (10.5,-5.5) -- (12,-7.5);
\draw[postaction={decorate}] (11.5,-8.5) -- (10,-6.5);

\draw[postaction={decorate}] (10.5,-8.5) .. controls (11.5,-9.2) and (12.7,-8.1) .. (12,-6.5);
\draw[postaction={decorate}] (11.5,-5.5) .. controls (9.8,-4.5) and (9.5,-6) .. (10,-7.5);
\draw[postaction={decorate}] (12,-6.5) .. controls (13.7,-9.1) and (10,-10) .. (10,-7.5);
\end{scope}

\node at (11,-9.25) {$\overrightarrow{H}_{16}$};

\end{tikzpicture}

\caption{A list of planar push cliques whose underlying graphs $H_1, H_2, \cdots, H_{16}$ is an exhaustive list of edge-minimal planar underlying push cliques.}\label{push_fig planar push cliques}
\end{figure}
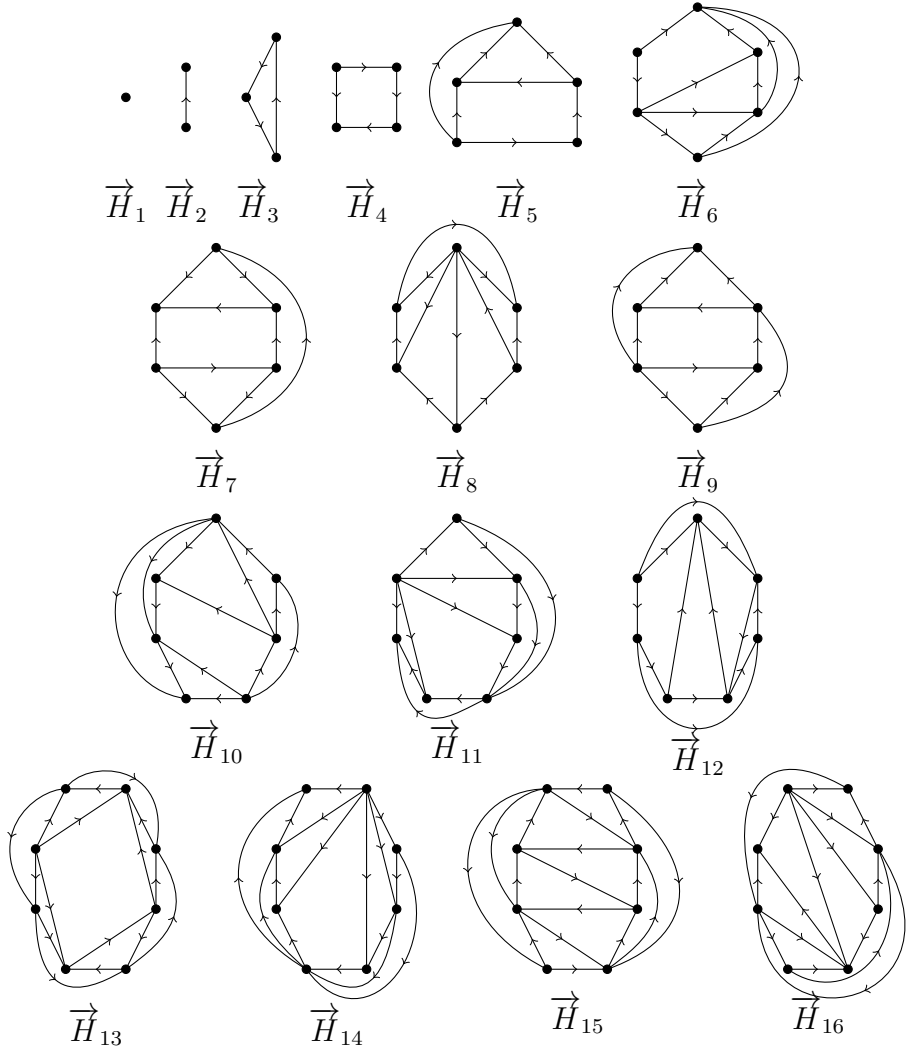

\begin{theorem}\label{push_th planar push clique list}
An undirected  planar graph  is an underlying push clique   if and only if it contains an underlying graph of one of the 16 planar graphs depicted in 
Figure~\ref{push_fig planar push cliques} as a spanning subgraph.
\end{theorem}

The proof of the above result is correct in the paper~\cite{bensmail2017oriented} except for some errors 
in Lemma~5.9 from~\cite{bensmail2017oriented}. The statement of the lemma is unchanged nevertheless. 
We recall the statement here along with a corrected proof.

\begin{lemma}\label{push_lem 7 list}
If $G$ is an underlying push clique having $|V(G)| = 7$,  then 
 $G$ contains   $H_{10}, H_{11}$ or $H_{12}$ as a spanning subgraph.  
\end{lemma}

\begin{proof}
We will try to construct a Hamiltonian planar reach-complete graph 
$G$ on 7 vertices,
without any dominating vertex,  
  with minimum degree 3, not containing 
 $H_{10}, H_{11}$ or $H_{12}$ as a subgraph.
  If such a graph $G$ does not exist, then we are done  due to 
  Lemma~5.8, Observation~5.3, Lemma~5.4 and Observation~2.3 from~\cite{bensmail2017oriented}.
  We will show that 
  such a graph $G$ does not exist through a case analysis.

Assume that $G$ is $C_7$ having $m$ medium chords  and $s$ short chords.

 If $m = 0$, then to make $G$ reach-complete we  need to have all the short chords,
 thus a $K_5$-minor.

If $m = 1$, then without loss of generality assume that the medium chord is $m_0$. 
Now we must add  $s_2$ and $s_6$ to make $G$ reach-complete. 
Also we add $s_5$ without loss of generality to have $d(a_5) \geq 3$.
  This graph has a Hamiltonian cycle 
$a_0a_3a_4a_2a_1a_6a_5a_0$   with 
2 medium chords $a_0a_1$ and $a_4a_5$. Thus this case gets reduced to the case $m \geq 2$.

If $m = 2$, then without loss of generality assume that one of the medium chords is $m_0$. The second medium chord can be chosen in three ways (up to symmetry):

\begin{itemize}
\item  The second medium chord is $m_4$. Observe that it is not possible to 
make this graph reach-complete by adding short chords without creating a dominating vertex or a $K_5$-minor.

 \item  The second medium chord is $m_5$. Thus we must have  $s_2$ 
 for 
$a_2$ to reach $a_4$ and $s_4$ for $a_4$ to reach $a_6$. 
This graph has a Hamiltonian cycle $a_0a_3a_4a_2a_1a_5$
$a_6a_0$ with 3 long chords $a_0a_1$, 
$a_4a_5$ and $a_4a_6$. Thus this case gets reduced to the  case $m \geq 3$. 
 
\item  The second medium chord is $m_1$. Without loss of generality we may add the short chord $s_0$ for having $d(a_2) \geq 3$. 
 Observe that there are exactly three ways  to 
make this graph reach-complete by adding short chords without creating a dominating vertex or a $K_5$-minor: 
$(i)$ by adding $s_0, s_2, s_3$ and $s_6$  implying $H_{12} \sqsubseteq G$, 
$(ii)$ by adding $s_0, s_2, s_4$ and $s_5$  implying $H_{11} \sqsubseteq G$, 
$(iii)$ by adding $s_0, s_3, s_5$ and $s_6$  implying $H_{12} \sqsubseteq G$.
\end{itemize}

 If $m = 3$, then three non incident medium chords will create a 
 $K_{3,3}$-minor. Thus we can assume that $G$ has two incident medium chords $m_0$ and $m_4$. Without loss of generality 
 the choice of the third chord gives us three subcases.

\begin{itemize}
\item  The third medium chord is $m_1$. Observe that it is not possible to 
make this graph reach-complete by adding short chords without creating a dominating vertex or a $K_5$-minor.

  \item  The third medium chord is $m_2$. 
 Observe that there are exactly two ways  to 
make this graph reach-complete by adding short chords without creating a dominating vertex or a $K_5$-minor: $(i)$ by adding $s_2$ and $s_6$  implying 
$H_{12} \sqsubseteq G$, $(ii)$ by adding $s_4$ and $s_6$  implying 
$H_{12} \sqsubseteq G$.

 \item  The third medium chord is $m_5$. 
 Observe that there are exactly two ways  to 
make this graph reach-complete by adding short chords without creating a dominating vertex or a $K_5$-minor: $(i)$ by adding $s_0$, $s_3$ and $s_6$  implying 
$H_{10} \sqsubseteq G$, $(ii)$ by adding $s_1$, $s_2$, $s_4$ and $s_6$  implying 
$H_{11} \sqsubseteq G$.
\end{itemize}

 If $m = 4$, then three non-incident medium chords will create a $K_{3,3}$-minor. 
 Thus barring those cases, without loss of generality we have two cases.

\begin{itemize}
\item  Suppose the four medium chords are $m_0, m_1, m_4$ and $m_5$. 
To have $d(a_2),$
$d(a_6) \geq 3$ we must add a matching of size 2 
in the set $\{a_1a_6, a_4a_6, a_0a_2, a_4a_2\}$ of short chords.  
Note that $s_5$ implies a dominating vertex $a_0$ and the edge $s_1$ implies 
a dominating vertex $a_1$. 
Suppose that we have $s_6$ and $s_0$. 
Then $a_3$ must reach $a_5$ by being adjacent $s_3$, creating a $K_5$-minor.

\item  Suppose the four medium chords are $m_0, m_1, m_3$ and $m_4$. 
Note that it is not possible to make this graph reach complete by adding short chords without creating a dominating vertex or  a $K_5$-minor.

\item  Suppose the four medium chords are $m_0, m_2, m_4, m_6$.  
To have $d(a_1) \geq 3$,  we must add one of the short chords from $\{s_1, s_6\}$. 
However, adding $s_1$ is symmetric (rotation after reflection)  
to adding $s_6$. 
Thus assume without loss of generality that we have added $s_6$ to our graph. 
Note that adding $s_1$ or $s_4$ will imply $H_{12}  \sqsubseteq G$. 
Observe that the only way to make this graph reach complete without adding $s_1$ or $s_4$ is to add the short chord $s_2$. This will imply $H_{10}  \sqsubseteq G$.
\end{itemize}

If $m \geq 5$, then $G$ is not planar as it contains a $K_{3,3}$.  
\end{proof}

\bigskip 

Moreover, we would like to report that there are, in total, 48 non-isomorphic 
planar underlying push cliques (1 on 1 vertex, 1 on 2 vertices, 1 on 3 vertices, 3 on 4 vertices, 4 on 5 vertices, 11 on 6 vertices, 14 on 7 vertices and 13 on 8 vertices). 
See the lists in the webpage:~http://jbensmai.fr/code/push/~for details.  

This count is updated and we have also updated the webpage link accordingly.

\bigskip

\noindent {\bf Acknowledgement:} The authors would like to thank Mr. Mithun Roy (Department of Computer Science and Engineering, 
Siliguri Institute of Technology, 
India) for pointing out the error.





\begin{thebibliography}{8}
\bibitem{bensmail2017oriented}
Bensmail, J. and Nandi, S. and Sen, S.: On oriented cliques with respect to push operation. Discrete Applied Mathematics \textbf{232}, 50--63 (2017)

\end{thebibliography}

%
%
\end{document}